\documentclass[a4paper,UKenglish]{article}

\usepackage{microtype}
\usepackage{todos}
\usepackage{amsmath}
\usepackage{amsfonts}
\usepackage{amssymb}
\usepackage{textcomp}
\usepackage{tikz}
\usepackage{url}
\usepackage{arydshln}
\usepackage{listings}
\usepackage{fullpage}
\newtheorem{definition}{Definition}

\usetikzlibrary{decorations.pathreplacing}

\usepackage{amsmath}
\usepackage{amsfonts}
\usepackage{amssymb}
\usepackage{xspace}
\usepackage{ifthen}
\usepackage{url}
\ifdefined\qed
\else
\newcommand*{\qed}{\hfill\ensuremath{\square}}%
\fi
\ifdefined\lemma
\else
\newtheorem{lemma}{Lemma}
\fi
\ifdefined\proof
\else
\newtheorem{proof}{Proof}
\fi
\ifdefined\theorem
\else
\newtheorem{theorem}{Theorem}
\fi

\usepackage{multirow}

\DeclareMathOperator{\pr}{Pr}

\usepackage{xspace}

\newcommand{\etal}{et al.\xspace}

\newcommand{\realrange}[2]{\left[#1, #2\right]}

\newcommand{\unitrange}[2]{\realrange{0}{1}}

\newcommand{\llabel}[1]{\label{\labelprefix:#1}}
\newcommand{\labelprefix}{} 

\newcommand{\discussionsize}{\small}

\newcommand{\notiz}[1]{}

\newenvironment{code}{\noindent
\begin{tabbing}%
\hspace{2em}\=\hspace{2em}\=\hspace{2em}\=\hspace{2em}\=\hspace{2em}\=%
\hspace{2em}\=\hspace{2em}\=\hspace{2em}\=\hspace{2em}\=\hspace{2em}\=%
\kill}{\end{tabbing}}

\newcommand{\labelcommand}{}
\newcommand{\captiontext}{}
\newsavebox{\codeparam}
\newcounter{lineNumber}
\newenvironment{disscodepos}[3]{%
\renewcommand{\labelcommand}{#2}%
\renewcommand{\captiontext}{#3}%
\sbox{\codeparam}{\parbox{\textwidth}{#3}}%
\begin{figure}[#1]\begin{center}\begin{code}\setcounter{lineNumber}{1}}{%
\end{code}\end{center}\caption{\llabel{\labelcommand}\captiontext}\end{figure}}

{\end{disscodepos}}

\newdimen\endofsize\endofsize=0.5em
\def\endofbeweis{~\quad\hglue\hsize minus\hsize
                 \hbox{\vrule height \endofsize width
\endofsize}\par}

\newcommand{\ignore}[1]{}

\title{Doing More for Less\footnote{This work was partially supported by DFG Grant 933.} -- {Cache-Aware Parallel Contraction Hierarchies Preprocessing}}
\author{Dennis Luxen and Dennis Schieferdecker\\Karlsruhe Institute of Technology\\
  Institute for Theoretical Computer Science\\ Karlsruhe, Germany\\
  \texttt{\{luxen,schieferdecker\}@kit.edu} }

\begin{document}

\maketitle

\begin{abstract}
Contraction Hierarchies is a successful speedup-technique to Dijkstra's seminal shortest path algorithm that has a convenient trade-off between preprocessing and query times. We investigate a shared-memory parallel implementation that uses $O(n+m)$ space for storing the graph and $O(1)$ space for each core during preprocessing. 
The presented data structures and algorithms consequently exploits cache locality and thus exhibit competitive preprocessing times. The presented implementation is especially suitable for preprocessing graphs of planet-wide scale in practice.
Also, our experiments show that optimal data structures in the PRAM model can be beaten in practice by exploiting memory cache hierarchies.
\end{abstract}

\section{Introduction and Related Work}\label{sec:related}

Computing point-to-point shortest (or fastest) path queries in a graph has been solved by Dijkstra's seminal algorithm \cite{d-ntpcg-59} since the early times of computer science.
A road network is modelled as a graph $G=(V,E)$ with $\vert V\vert=n$ nodes and $\vert E\vert=m$ edges.
Each edge $e\in E$ is associated with a cost $c(e)$ that is required to traverse that edge.
For the sake of simplicity, consider that nodes are identified by their ID, i.e. $a\in V$ is treated as a number if appropriate.

While the running time of Dijkstra's algorithm is clearly polynomial, its running time does not scale to large instances, e.g. road networks of continental size.
A well-tuned implementation still needs a few seconds for a single shortest path query on such a network even on today's hardware.

Heuristics to prune the search space provide a sense of goal direction \cite{hnr-afbhd-68,gh-cspas-05}.
At one point, the algorithm engineering community picked up the problem and started providing so-called \textit{speedup-techniques} to Dijkstra's algorithm that deliver much better query performance.
An early technique with substantial speedup and optimal (guaranteed) shortest paths is called Arc-Flags \cite{l-spgef-97,mssww-pgsda-06}, where the graph is partitioned into regions and each edge stores a flag for each region if there exists some shortest paths over it into the respective region.
The authors refer the interested reader to \cite{dssw-erpa-09} for a survey on a number of route planning techniques.

One of the optimal speedup-techniques is Contraction Hiearchies (CH) \cite{gssv-erlrn-12}.
CH has a convenient trade-off between preprocessing and query time and exploit the inherent hierarchy of a road network.
CH \textit{shortcut} all nodes of the graph in some order.
Here, shortcutting means that a node is (temporarily) removed from the network and replaced by as few shortcut edges as possible to preserve shortest path distances.
A so-called \textit{witness search}, which is a unidirectional Dijkstra run, is applied to check if a shortcut is actually necessary.
The union of the set of original edges and and the set of shortcut edges form a directed acyclic graph (DAG).
A CH query is essentially a bidirected Dijkstra query and needs only to relax an edge when the target node was contracted after the starting node.
The number of settled nodes during a query is in the order of a few hundred nodes and the query times are about 100 microseconds.
The fastest CH variant is CHASE \cite{bdsssw-chgds-10}, where CH is combined with Arc-Flags and queries run in the order of ten microseconds.

The priority function $p(\cdot)$ by which the nodes are ordered to be contracted is heuristic.
Its purpose is to reflect which nodes are more \textit{important} than others, i.e. the node reflecting a junction with high traffic is said to be more important than the node modelling the end of a dead-end street in a quiet neighborhood.
Usually, it is a simulated witness search that (among other things) inspects the number of edges that would have to be inserted if a certain node would have been removed.
For example, the priority functions of more or less all known CH implementations use a very local search only for the sake of preprocessing efficiency, e.g. \cite{gssv-erlrn-12,bdsv-tdch-09,adgw-ahbla-11,adgw-arrn-10}. Here, local means that only a $k$-neighborhood around a node is considered to determine its priority. 
Batz \etal \cite{bgns-tdcha-10} and Batz and Sanders \cite{bs-tdrpg-12} explore a 16-hop neighborhood, while Kieritz \etal \cite{klsv-dtdch-10}  consider a $5$-hop neighborhood in a distributed memory parallel implementation of CH.
Other implementations \cite{lv-rtros-11}, including the one at hand, prune the search space by setting a fixed limit on the number of settled nodes during the witness search.

The \textit{Parallel Random Access Machine (PRAM)}, e.g. \cite{Karp1990}, is a simple model of parallel computation.
A PRAM has a global (shared memory) and $p$ processors, each equipped with a private local memory.
Each processor can access either shared or global memory in unit time, as well as perform a computation with respect to a memory access.
The cost of accessing memory is uniform for all processors and all accessible memory locations.
The PRAM model is simple and easy to understand but does not reflect the several levels of on-chip cache memory that are present in modern CPUs.
Arge \etal \cite{Arge2008} proposed the \textit{parallel external memory (PEM)} model to capture the memory hierarchies of modern processor architectures.
This model is \textit{cache-aware} and the authors present how to conduct a number of fundamental operations like prefix sums and sorting efficiently.

\textit{Tabulation hashing} is a simple hashing scheme that dates back to as early as the late 1960s when first published by Zobrist \cite{z-anhma-69} and the late 1970s when rediscovered by Carter and Wegmann \cite{cw-u-79}.
It uses simple table lookups and \textit{exclusive or (\texttt{XOR})} operations.
Later Patrascu and Thorup \cite{Patrascu2011} gave a theoretical analysis of the scheme.
Tabulation hashing interprets input keys as a string of $c$ characters $x_1,\ldots,x_c$.
For each of the possible character positions a random table $T_i, 1\leq i \leq c$ is initialized and the following hash function is used:

$$ h(x)=T_1[x_1] \oplus \ldots \oplus T_c[x_c] $$

\subsection{Parallel Preprocessing of Contraction Hierarchies}
Vetter \cite{v-ptdch-09} proposed a parallel preprocessing algorithm that identifies nodes, which can be contracted in parallel.
The method gives good speedups until the memory bandwidth is saturated.
More formally, an \textit{independent node set} is a set of nodes that can be contracted independently of any other remaining nodes in the graph.
Every node that is of lowest priority within a neighborhood of $k$ hops is added to the independent set, where $k$ is tuning parameter.
While $k$ is a tuning parameter, notice that $k=2$ is sufficient.
Kieritz \etal \cite{klsv-dtdch-10} generalize this approach to a distributed memory setting where nodes are contracted on separate compute nodes of a cluster and graph changes are only communicated when necessary.

Since the priority of two nodes may be equal, it is necessary to install a tie-breaking rule.
Note that the actual order in which the nodes of an independent set are contracted can be arbitrary.
Although, Vetter \cite{v-ptdch-09} verbalizes the need for a tie-breaking mechanism, it is not further specified.

\subsection{Our Contribution}
The contribution of this paper is twofold.
First, it shows how well-chosen data structures and algorithms deliver better perfomance than others that are optimal in the PRAM model of computation by exploiting caching effects.
Second, it describes a shared-memory parallel implementation of CH that uses only constant space per CPU core.

The remainder of this paper is structured as follows.
Chapter \ref{sec:tie-break} explains the tie-breaking mechanism more thoroughly and shows basic properties.
Subsequently, a tie-breaker based on tabulation hashing with theoretical perfomance guarantees is developed.
The experimental evaluation shows that it pays off to invest into executing more processor instructions when a large number of cache faults can be avoided.
Building on theoretical performance guarantees, Section \ref{sec:heap-storage} generalizes this hashing technique to build a key-value storage for the priority queue that is used during the witness searches of the CH preprocessing.
Sections \ref{sec:tie-break} and \ref{sec:heap-storage} show the performance in practice while Section \ref{sec:conclusion} gives concluding remarks and identifies future work.
To the best of the authors' knowledge this is the first work that describes a sophisticated shared-memory parallel implementation of CH with an emphasis on cache-awareness.

\section{Tie-Breaking using Tabulation Hashing}\label{sec:tie-break}

As mentioned briefly in the related work section, the role of tie-breaking is to facilitate the decision which node to contract only when neighboring nodes have equal priorities.
A tie-breaking mechanism  cannot be an arbitrary decision process but has to fulfill certain properties as the following paragraph shows.

\begin{definition}[Node Ordering and Tie-Breaking]
Consider two nodes $u\neq v \in V$.
A node $u$ is \textit{smaller} than node $v$ from the $k$-neighborhood if $p(u) < p(v)$ or if $u \prec v$ in case $p(u) = p(v)$, where $\prec$ defines an order on the nodes. The order (or tie-breaker) is called \textit{consistent} iff. $u \prec v = \neg (v\prec u)$.
\end{definition}

One can show that the property of consistency is an essential property of any correct CH implementation.
Consider the contraction of a single node to be a basic operation during the preprocessing.

\begin{lemma}
CH preprocessing with an inconsistent tie-breaker does not terminate for all inputs.
\end{lemma}
It suffices to show that there exists an input graph and an inconsistent tie-breaker for which no node is selected during an iteration.
\begin{proof}
Consider a triangle of three nodes $a,b,c$ each of degree two with equal priority.
Further assume that the tie-breaker is inconsistent with $x\prec y = 0$, $\forall x,y\in \{a,b,c\}$.
No node will be selected to be an element of the independent set that is to be contracted.
Thus, the contraction does not terminate, i.e. is not wait-free for all inputs.
\end{proof}
\subsection{Simple Tie-Breaking}

The easiest implementation of a consistent tie-breaker is a \textit{random shuffle} of the node IDs and a subsequent renumbering of the graph.
A random shuffle of node IDs implies linear work in the number of nodes and edges.
From a theoretical point of view, one could argue that this is as good as it gets since the work is constant per decision.
On the other hand, the constants associated with such a scheme may render it impractical.
Doing a random shuffle and a subsequent renumbering on the nodes of a large graph, e.g. for the entire planet, can be prohibitively expensive in practice.
Preliminary experiments preceding this paper showed that this can take as long as contracting the first 20-25\% of the nodes.
Furthermore, a disadvantage of random shuffling is that it breaks any inherent cache-efficiency that the data has.
In real-world data sets node IDs are given in the order in which they are created, i.e. consecutive numbering of the nodes of an entire street when it is added into the data set.
There are conflicting interests for the numbering of the nodes.
On one hand, the strength of CH is that its data structure is quite different from the intuition of a hierarchy of road types.
And thus, one would want a preprocessing that is independent from any existing ordering or presentation of the input data.
On the other hand, the preprocessing mostly consists of small graph searches and one would like the data to display a certain amount of locality, i.e. close-by nodes have close-by IDs, to leverage cacheing effects.

The simplest data structure that has (PRAM) optimal query time of $O(1)$ to implement the tie-breaking with a \textit{bias array} of size $n$, e.g. in the implementations of \cite{klsv-dtdch-10,v-ptdch-09}.
An array $A$ is populated with numbers $0,\ldots,n-1$ and randomly shuffled at the beginning of the preprocessing.
This yields a precomputed pairwise distinct random number for each node $v$ in the graph.
When a tie-break is necessary for nodes $i$ and $j$ then the values of $A[i]$ and $A[j]$ are compared.
The memory overhead is $\mathcal{O}(1)$ space per element, which is acceptable from a practical as well as from a theoretical point of view.
The main advantage of this rule is its simplicity and that it desirably preserves the locality of any existing numbering.
This is formalized by the following 

\begin{definition}[Independence from Input Numbering]\label{def:tie-breaker-independence}
A tie-breaking ordering is called \textit{ID-independent}, or short \textit{independent}, if its outcome is irrespective of any input numbering. Likewise, an ordering is said to be $(1-\varepsilon)$-independent from the node ordering if the probability that $u\prec v$ is an independent and identically (i.i.d.) random choice is larger or equal to $(1-\varepsilon)$,

$$\pr[(u\prec v)=\text{ i.i.d. random choice}]\geq (1-\varepsilon).$$
\end{definition}

The above straight-forward implementation of tie-breaking has one major disadvantage in practice that is a direct result of its simplicity.
While one expects this tie-breaker to be fast, the number of cache misses is large. 
Contrary to the PRAM model and somewhat following along the lines of the more realistic PEM model, memory accesses are not uniform in reality.
Generally speaking, access to a small local cache is fast while accessing the shared memory is quite expensive.
The bias array is much larger than any cache size even for medium-sized graphs and one must expect an expensive cache miss for each call to the tie-breaking rule, even if the data exhibits some locality preserving node numbering.

A preliminary experiment with a memory debugging tool revealed that most of the accesses to the bias array were actually cache faults.
While literature is generally scarce on the subject, the number of clock cycles wasted in a cache miss easily amount to a few hundred \cite{Drepper2007}.

\subsection{A Fast $(1-\varepsilon)$-Independent Tie-Breaker}

The following hashing-based scheme gives the basis of a tie-breaking mechanism that takes constant time to evaluate and uses constant space only.
It is not only independent with high probability, but surprisingly fast in practice and even faster than the above simple schemes.

To build a tie-breaker for two nodes $a,b\in V$ the hash values of $a$ and $b$ are compared and in the (unlikely) event that they are equal, $a$ and $b$ are compared directly.
More formalized
\begin{definition}[Tie-Breaking by Tabulation Hashing]\label{def:tie-break-order}
Given a (tabulation) hash function $h:U\rightarrow [m]$ and two elements $a,b\in U$, then the boolean expression
$$a\prec b := \left[h(a)<h(b)\right]\vee\left[(p(a)\equiv p(b))\wedge(a<b)\right] $$
obviously defines an order on the elements of $U$.
\end{definition}

\subsection{Analysis of Performance Guarantees}\label{sec:guarantees}

The following analysis gives performance guarantees for the tabulation hash based tie-breaker and leads to showing the following lemma:
\begin{lemma}[Perfomance Guarantees]\label{lemma:tie-break-guarantees}
Tabulation hash based tie-breaking uses sublinear space, evaluates in constant time and is $(1-\varepsilon)$-independent for $\varepsilon > 0$.
Furthermore, the resulting ordering is consistent.
\end{lemma}

\begin{proof}
The first two properties follow from the construction of the data structure.
To show the third property, it is necessary to show that the expected fraction hash collisions is bounded above by $\varepsilon$.
The analysis of the tabulation hash based tie-breaker builds on an earlier result of Carter and Wegmann \cite{cw-u-79} and the definition of $k$-independent hashing:

\begin{definition}[$k$-Independent Hashing]
A family of hash functions $H=\{h:U\rightarrow [m] \}$ is said to be $k$-independent if randomly selecting a function $h\in H$ guarantees for $k$ distinct keys $x_1,\ldots,x_k \in U$ and $k$ hash codes $y_1,\ldots,y_k$ that

$$\pr\limits_{h\in H}\left[h(x_1)=y_1 \wedge\ldots\wedge h(x_k)=y_k\right]\leq m^{-k}.$$
\end{definition}

A property that directly results from this definition is the fact that for fixed keys $x_1\ldots,x_k\in U$ and a randomly drawn hash function $h\in H$, the hash values $h(x_1),\ldots,h(x_k)$ are independent random numbers.
Carter and Wegmann \cite{cw-u-79} show that tabulation hashing is $3$-independent.
Thus, the probability of a hash collision is less than $2^{-k}$ and thus the order it defines is random with high probability.
Only in the rare case, of a collision the ordering is derived from the IDs of the nodes.
Note that plugging the previous result into the above construction directly establishes the $(1-\varepsilon)$-independence.
Actually, the previous result by Carter and Wegmann is even stronger than necessary as a $2$-independent ($1$-universal) hashing scheme would have  sufficed.

To show the last claim of consistency, consider two node IDs $a$ and $b$ for which the ordering is determined.
It suffices to show that $a\prec b = \neg (b\prec a)$.
To the contrary, consider that the tie-breaker evaluates $a\prec b$ and also $b\prec a$ to true.
The tie-breaker either evaluates the hash values $h(a)$ and $h(b)$ or $a$ and $b$ directly if hash values are equal.
In both cases the above contraposition leads to a contradiction.
This concludes the proof of Lemma \ref{lemma:tie-break-guarantees}.
\end{proof}
\subsection{The Actual Implementation}

The implementation splits the 32-bit sized input ID of any node into two words of size 16 bit.
Thus, two lookup tables with $2^{16}$ entries have to be filled with pairwise distinct random numbers.
This is done by filling the tables consecutively with numbers $0,\ldots,2^{16}-1$ and then applying a random shuffle.
The overhead of initializing these arrays is neglectable, since this has to be done only once and the associated work is linear in the  size (and number) of the lookup tables.

A query is straight-forward.
Input ID $x$ is split into the most and least significant halves, a lookup is performed for each of the sub words and then combined by a \texttt{XOR} operation.
Note that the work necessary to perform a query is constant.
See Figure \ref{fig:tie-break} for an illustration of the implementation of this scheme.

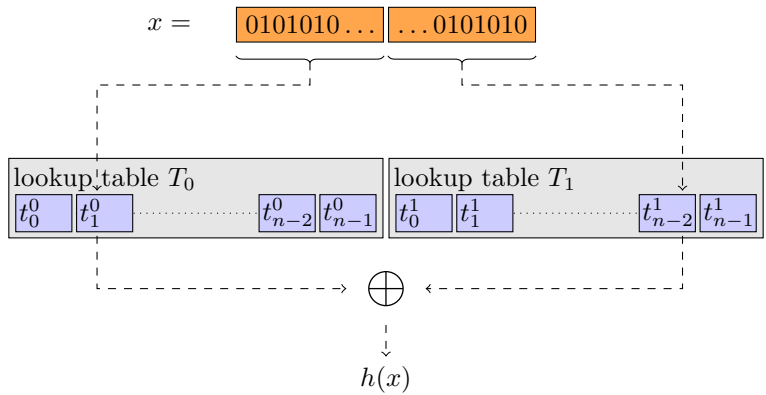
\begin{figure}[htb]
\begin{center}
\scalebox{1.0}{
\begin{tikzpicture}

\tikzstyle{block1} = [draw,fill=blue!20,minimum width=2.1em,minimum height=1.4em]
\tikzstyle{block2} = [draw,fill=orange!70,minimum width=2.8em]
\tikzstyle{pointer}= [->,thick,gray,dashed,step=1pt]

\node[draw,minimum height=3em,minimum width=14em,fill=gray!20,overlay] at (4,-1.8) {};
\node[draw,minimum height=3em,minimum width=14em,fill=gray!20,overlay] at (-1,-1.8) {};

\draw[dotted] (2,-2) -- (6,-2);
\draw[dotted] (-3,-2) -- (1,-2);

\node at (-2.2,-1.55) {lookup table $T_0$};
\node at (2.8,-1.55) {lookup table $T_1$};

\foreach \x/\y in {0/0,1/1,4/n-2,5/n-1} {
   { \node[block1] at (2+\x*0.8,-2) (kasten\x) {}; }
   { \node[block1] at (-3+\x*0.8,-2) (kasten\x) {}; }
   { \node[right] at (1.55+\x*0.8,-2) (text\x) {$t_{\y}^1$}; }
   { \node[right] at (-3.45+\x*0.8,-2) (text\x) {$t_{\y}^0$}; }
}

\node at (-1.35, 0.5) (x) {$x=$};

\node[block2] at (0.5,0.5) (msb) {$0101010\ldots$};
\node[block2] at (2.5,0.5) (lsb) {$\ldots0101010$};

\draw [decoration={brace,mirror,raise=5pt},decorate] (lsb.south west) --  node[above=12pt]{}(lsb.south east);
\draw [decoration={brace,mirror,raise=5pt},decorate] (msb.south west) --  node[above=12pt]{}(msb.south east);

\node[scale=2,overlay] at (1.5,-3) (oplus) {$\oplus$};

\path[draw,dashed,->] (0.5,-0.0) |- (0,-0.3) -| (0.-2.3,-1.7);
\path[draw,dashed,->] (2.5,-0.0) |- (2.5,-0.3) -| (5.4,-1.7);

\path[draw,dashed,->] (-2.3,-2.3) |- (oplus.west);
\path[draw,dashed,->] (5.4,-2.3) |- (oplus.east);


\path[draw,dashed,->] (oplus.south) -- (1.5,-3.9);

\node at (1.5,-4.2) {$h(x)$};

\end{tikzpicture}
}
\end{center}
\caption{Sketch of the Tabulation Hash Function.}
\label{fig:tie-break}
\end{figure}

The above hashing data structure can then be combined to the following tie-breaking algorithm by implementing Definition \ref{def:tie-break-order}.
Consider the following code fragment of Listing \ref{code:branching-tie-break}

\begin{lstlisting}[label=code:branching-tie-break,caption=Tabulation-based Tie-Breaker]
bool bias(const NodeID a, const NodeID b) {
    unsigned short hasha = h(a);
    unsigned short hashb = h(b);

    if(hasha != hashb)
        return hasha < hashb;
    return a < b;
}
\end{lstlisting}

The number of expected collisions is tiny as shown in the analysis.
The observed rate of collisions in practice is less than $10^{-5}$.
The entire tie-breaking mechanism, including hashing, uses as few as 22 assembly instructions on an X86 CPU in practice, when letting GCC optimize the code (\texttt{-O3} flag).
See Appendix \ref{apx:assembly} for the actual assembly listing.
Most interestingly, it is possible to evaluate the \texttt{if}-statement without any branching by using conditionally set flags in the register\footnote{X86 assembly instruction \texttt{setg}}.
The space requirement for this tie-breaking mechanism is 256kb of RAM, which fits into the L2 cache of any recent X86 processor. 
As mentioned before, the literature on processor cache timings is scarce, but L2 cache latency is approximately ten cycles.
Table \ref{tab:tie-breaker-results} gives the results of experiments of running times of CH preprocessing either with bias-array based tie-breaking or tabulation hash-based tie-breaking.

\subsection{Experimental Evaluation}\label{sec:exp-tiebreak}

The experiments to evaluate the practical impact of this tie-breaking scheme have been evaulated on $8$ cores of an AMD Opteron 6212 clocked at 2.6 GHz running Linux kernel version 3.0.0 and 128 GB of RAM.
The processor has $8 \times 16$ KBytes of L1 data caches, $4 \times 2$ MB shared exclusive L2 caches and 16Mbytes of L3 cache memory.
The datastructures and algorithms were implemented in C++ and compiled with GCC 4.6.1 using full optimizations (\texttt{-O3}).
The graph instances represent road networks of various sizes ranging from the metropolitan area of Berlin, Germany to the planet-wide data of the OpenStreetMap\footnote{\url{http://www.openstreetmap.org}} Project as of July, 4th, 2012.

\begin{table}[htb]
\begin{center}
\begin{tabular}{lrr}
 & \multicolumn{2}{c}{Graph Size}\\
\cline{2-3}
 & \multicolumn{1}{c}{$\vert V\vert$} & \multicolumn{1}{c}{$\vert E \vert$} \\ 
Berlin   & 288\,755      &         844\,550 \\ 
Baden    & 5\,108\,952   &     12\,829\,610 \\ 
Germany  & 33\,927\,089  &     86\,477\,642 \\ 
Planet   & 758\,206\,383 & 1\,842\,527\,702 \\ 
\end{tabular} 
\caption{Graph Sizes of several Road Networks used in the experimental evaluation.}
\label{tab:instance-sizes}
\end{center}
\end{table}

The graphs used in the experiments are \textit{edge-expanded}, i.e. each possible turn is explicitly modelled and U-turns are forbidden \cite{v-rprnt-08}.
Moreover, existing turn restrictions present in the input data are preserved.
Note that the CH preprocessing time is higher for edge-expanded networks than for unexpanded graphs as previously observed by Delling \etal \cite{dgpw-crp-11}.
Table \ref{tab:instance-sizes} gives the basic properties of the road networks after edge-expansion.

Table \ref{tab:tie-breaker-results} reports on the impact of the hashing scheme on the duration of the preprocessing.
Columns \textit{bias array} and \textit{xorhash} denote the preprocessing times for a bias-array based tie-breaker and for the tabulation hash based tie-breaker, while column \textit{speedup} denotes the observed speedup.
\textit{saving} indicates the amount of memory saved by the tabulation hashing scheme over the bias array.
\begin{table}[htb]
\begin{center}
\begin{tabular}{lrrrr}
& \multicolumn{2}{c}{duration [s]} & & [GiB]\\
\cline{2-3}\cline{5-5}
        & bias-array & xorhash & speedup & saving \\ 
\hline
Berlin  &      17.30 &      12.06 & 1.43 & $8.9\cdot 10^5$\\ 
Baden   &     105.06 &      80.29 & 1.30 & $2.1\cdot 10^7$\\ 
Germany &     827.48 &     628.96 & 1.31 & $1.3\cdot 10^8$\\ 
Planet  & 21\,873.58 & 15\,815.10 & 1.38 & $3.0\cdot 10^9$
\end{tabular}
\caption{Experimental results for the tabulation-based tie-breaking scheme}
\label{tab:tie-breaker-results}
\end{center}
\end{table}

The experiments show that a tie-breaking mechanism based on tabulation hashing not only reduces the memory requirements, but also that it pays off to trade some processing cycles for much better cache efficiency.
The benefits of tabulation hash based tie-breaking are twofold.
The speedup is consistently between 25--30\% and the space requirement is constant.

An extended profile run was conducted on a smaller edge-expanded graph resembling the street network of Berlin instance from Table \ref{tab:instance-sizes}.
This was done using the \textit{cachegrind} plugin of Valgrind\footnote{\url{http://www.valgrind.org}}, a tool for (memory usage) debugging and profiling.
Examining larger instances is impractical since the tool entirely simulates the cache hierarchy of a modern processor, which takes orders of magnitude longer than running on real hardware.
However, the experiment revealed that while the overall instruction count  increased but slightly, the number of (simulated L1 and LL) cache misses dropped significantly by more than 20\%.
The overall number of executed instruction rose by less than 1\%, which again shows the low computational overhead of tabulation hashing.

\section{Heap Storage using Tabulation Hashing}\label{sec:heap-storage}

The application of tabulation hashing is not limited to implemented of a tie-breaking with guarantees for independent set generation.
The parallel preprocessing needs a priority queue per thread.
The data structure used to implement the priority queue is a binary heap that stores its content in a table.
The standard implementation is an array of size linear in the number of nodes of the graph.
This Section shows that the simplicity and formidable cacheing behavior explored in Section \ref{sec:tie-break} make tabulation hashing a great candidate to construct a hash table from.

As already mentioned in the related work of Section \ref{sec:related} the subgraphs that are explored during witness searches are rather small.
The implementation used for this paper prunes these searches at 1\,000 nodes for simulation and at 2\,000 nodes for actual contractions.

As the analysis of Section \ref{sec:guarantees} shows, the probability of a hash collision is small for tabulation hashing.
Also, the range of the hash function is $2^{16}$ and is of much larger cardinality than the set of at most 2\,000 explored nodes.
Thus, it is worthwile to implement a hash table using tabulation hashing that can be used as a storage table for the priority queue implementation.

It is necessary to use a collision resolution strategy, since a hash function points only to a records location and not to the record itself.
It seems obvious to use linear probing as resolution strategy for two reasons.
First, the number of collisions is small and so is the expected number of cells in the hash tables that have a non-vacant neighbor.
Second, the next cells are very likely to lie in the same cacheline as the original cell and therefore accesses to it are virtually cost-free.

\subsection{The Actual Implementation}
The implementation is mostly straight-forward.
A hash value is generated for each input key and its value is stored at the corresponding hash cell.
Collision, i.e. when the cell is not empty, are resolved by linear probing.
After each witness search the storage table of the priority queue is reinitialized.
While this seems non-obvious at first sight, one has to pay special care for the reinitialization of the storage array.
Resetting an array to initial values is expensive as it either involves a reallocation or a sweep over the memory or even both.
Therefore each cell has a local timestamp that indicates the time when was written last.
Initially, the global timestamp is zero and incremented each time the storage table is cleared.
This way, it is not necessary to actually zero out any memory, and it suffices to do a simple comparison during collision resolution.

The implementation uses 4 bytes each for key and value as well as for the timestamp which yields cell sizes of 12 bytes and therefore an overall memory consumption of 384 kilobytes per queue for the storage table.
Note that the table has only $32\,768=2^{15}$ entries which is half the range of the hash function.
Experiments showed that the collision rate was virtually unaffected while the memory consumption further decreased by 50\%.

\subsection{Experimental Evaluation}\label{sec:exp-storage}

The experiments on the performance of the tabulation hash based heap storage have been evaluated on $8$ cores of an AMD Opteron 6212 clocked at 2.6 GHz running Linux kernel version 3.0.0 and 128 GB of RAM again.
The datastructures and algorithms were implemented in C++ and compiled with GCC 4.6.1 using full optimizations (\texttt{-O3}).
The graph instances resemble the same edge-expanded road networks of various sizes also used in the experiments of Section \ref{sec:tie-break}.

The implementation of CH already includes the tabulation hash based tie breaker from above.
Preprocessing was run for the same instances as before and compared against two \textit{standard} hash table implementations.
The first one is a hash table implementation  from the Boost\footnote{\url{http://boost.org}} C++ library version 1.4.6, namely \texttt{boost::unordered\_map}.
This hash table is said to be close to the implementation of GCC C++0x hash table implementation.
The second implementation is Google's sparsehash\footnote{\url{http://code.google.com/p/sparsehash}} library version 1.10, namely \texttt{google::dense\_hash\_map}.
This hash table has the reputation of being among the fastest hash table implementations.

Table \ref{tab:hash-compare} gives the results from the experiments on a number of input graphs, where \textit{xorbreak} denotes the implementation of Section \ref{sec:tie-break} and \textit{xorhash} denotes the implementation of this Section.
The reference values of the plain bias-array implementation are given in line \textit{bias}.
Note that this variant also uses an array based storage for the priority queue.
Best values are printed in {\bf bold font}.
Note that preprocessing the planets road network did not complete within 18 hours for the Boost and Google hash table implementations.
Column \textit{Bytes Per Core} gives the overhead of the data structures per core that is used during preprocessing.
Reliable overhead values could not be retrieved and therefore left out from the comparison.

\begin{table}[htb]
\begin{center}
\begin{tabular}{lrrrrrrr}
 &  & \multicolumn{4}{c}{duration [s]} & & \multicolumn{1}{c}{Bytes}\\
\cline{3-6}
 & & Berlin & Baden & \multicolumn{1}{c}{Germany} & \multicolumn{1}{c}{Planet} & & per core \\ 
\cline{1-1}\cline{3-6}\cline{8-8}
  boost      &  & 49.51 & 249.29 & 1876.40 & \multicolumn{1}{c}{--} & & \multicolumn{1}{c}{--} \\ 
  google     &  & 26.66 & 165.73 & 1215.67 & \multicolumn{1}{c}{--} & & \multicolumn{1}{c}{--} \\ 
\hdashline
  bias-array &  & 13.74 & 102.09 &  822.21  &   21\,873.58  & & $n\cdot 4$ \\
\hdashline
  xorbreak   &  & {\bf 12.06} & {\bf 80.29} &  {\bf 628.96} & {\bf15\,815.10} & & $n\cdot 4$ \\
  xorhash    &  &      17.30  &     105.06  &        827.48 &     16\,030.90  & & \textbf{384k}
\end{tabular}
\end{center}
\caption{Preprocessing Results for Several Hash Storage}
\label{tab:hash-compare}
\end{table}

Most notably, the performance of the \textit{xorbreak} is the fastest among all experiments.
The gap between \textit{xorhash} and \textit{xorbreak} decreases as the road networks grow in size.
For the planet data set, the relative difference is as small as $2\%$.
An explanation for this is that cache faults occur more often the larger the storage table of the priority queue gets.
The memory consumption of xorhash is only constant per core.
Hence, the number of cache faults occuring in the \textit{xorhash} variant is robust against the size of the graph.

\section{Concluding Remarks and Future Work}\label{sec:conclusion}
We presented an algorithmic tuning parameter between preprocessing efficiency and space requirements.
Speaking in generalized terms, applying tabulation to tie-breaking gives a reasonable speedup in preprocessing efficiency that gives room for further optimization on the space required during preprocessing.

The high performance of the tabulation hashing applications can be attributed to much better cache locality of the intermediate data structures.
This locality has been leveraged during data structure design and implementation.
Carefully chosen and engineered data structures and associated algorithms allow for much flexibility during the preprocessing of large real-world road network instances.
For example, if speed is of essence and memory available then only the tabulation hash based tie-breaking may be applied while in a setting where memory is tight it may be both.

We showed a consequent application of tabulation hashing in CH preprocessing not only gives data structures that have constant size per core, but also preprocessing performance that is en par with previous implementations that use the theoretical best data structures in the PRAM model.

Furthermore, we would like to investigate the application of succinct graph data structures to bring the space requirement of preprocessing closer to the information theoretical lower bound while enjoying competitive preprocessing and query times.

\section*{Acknowledgements}
The authors would like to thank Peter Sanders for bringing up tabulation hashing in the first place and Nodari Sitchinava for great discussions on the topic.

\bibliography{../references,../references-chPreproc}

\appendix
\section{X86 Assembly Code of Tabulation Hashing}\label{apx:assembly}
\begin{lstlisting}[language={[x86masm]Assembler},label=code:-asm-tie-break,caption=X86 Assembly Tie-Breakers Bias function, backgroundcolor=\color{lightgray}]
 movq    table2(%rip), %rdx
 movl    %edi, %eax
 movzwl  %di, %edi
 shrl    $16, %eax
 movzwl  %ax, %eax
 movzbl  (%rdx,%rax), %eax
 movq    table1(%rip), %rdx
 xorb    (%rdx,%rdi), %al
 movzbl  %al, %eax
 ret
\end{lstlisting}

\end{document}